\documentclass[a4paper,UKenglish,cleveref, autoref]{lipics-v2019}
\usepackage{subcaption}
\usepackage{cleveref}
\usepackage{amsthm}
\usepackage{hyperref}

\usepackage{amssymb,amsmath}
\usepackage[backgroundcolor=white,linecolor=black,textsize=scriptsize]{todonotes}
\usepackage{cite}
\usepackage{xspace}
\usepackage{rotating}
\usepackage{xcolor}
\usepackage{colortbl}

\usepackage{lineno}
\usepackage{footnote}
\usepackage{thm-restate}
\usepackage{cleveref}

\makesavenoteenv{tabular}

\linenumbers

\usepackage{marcel_notation}

\newcommand{\RCMRunningTime}{$O\left( (kn + m)^2 \log\left(kn +
m\right)\right)$\xspace}

\newcommand{\dimacs}{\textsc{Dimacs}\xspace}
\newcommand{\sparsemc}{\textsc{Sparse MC}\xspace}
\newcommand{\ogdf}{\textsc{Ogdf}\xspace}
\newcommand{\stress}{\textsc{Stress}\xspace}
\newcommand{\SN}{\ensuremath{\mathcal S_0}\xspace}
\newcommand{\SF}{\ensuremath{\mathcal S_{512}}\xspace}
\newcommand{\RN}{\ensuremath{\mathcal R_0}\xspace}
\newcommand{\RF}{\ensuremath{\mathcal R_{512}}\xspace}
\newcommand{\WF}{\ensuremath{\mathcal W_{512}}\xspace}

\definecolor{cellbg}{rgb}{0.88,0.89,0.93}
\definecolor{cellbg2}{rgb}{0.92,0.96,0.88}
\newcommand{\good}{\cellcolor{cellbg2}}
\newcommand{\bad}{\cellcolor{cellbg}}

\title{Geometric Crossing-Minimization -- A Scalable Randomized Approach}

\titlerunning{Geometric Crossing Minimization}

\author{Marcel Radermacher}{Department of Computer Science, Karlsruhe  Institute of Technology,
	Germany}{radermacher@kit.edu}{}{}

\author{Ignaz Rutter}{Department of Computer Science and Mathematics, University
	of Passau, Germany}{rutter@fim.uni-passau.de}{}{}

\authorrunning{M. Radermacher and I. Rutter}

\Copyright{Marcel Radermacher and Ignaz Rutter}



\funding{Work was supported by grants RU 1903/3-1 and WA 654/21-1  of the German Research Foundation(DFG). }

\nolinenumbers 

\hideLIPIcs  

\EventEditors{Michael A. Bender, Ola Svensson, and Grzegorz Herman}
\EventNoEds{3}
\EventLongTitle{27th Annual European Symposium on Algorithms (ESA 2019)}
\EventShortTitle{ESA 2019}
\EventAcronym{ESA}
\EventYear{2019}
\EventDate{September 9--11, 2019}
\EventLocation{Munich/Garching, Germany}
\EventLogo{}
\SeriesVolume{144}
\ArticleNo{73}

\begin{document}



\maketitle

\begin{abstract}
	We consider the minimization of edge-crossings in geometric drawings of graphs
	$G=(V, E)$, i.e., in drawings where each edge is depicted as a line segment.
	The respective decision problem is $\cNP$-hard~\cite{bienstock1991some}.
	Crossing-minimization, in general, is a popular theoretical research topic;
	see Vrt'o~\cite{crBibVrto}. In contrast to theory and the topological setting,
	the geometric setting did not receive a lot of attention in practice. Prior
	work~\cite{doi:10.1137/1.9781611975055.12} is limited to the
	crossing-minimization in geometric graphs with less than $200$ edges.  The
	described heuristics base on the primitive operation of moving a single
	vertex $v$ to its \emph{crossing-minimal position}, i.e., the position in
	$\R^2$ that minimizes the number of crossings on edges incident to $v$. 

	In this paper, we introduce a technique to speed-up the computation by a
	factor of $20$. This is necessary but not sufficient to cope with graphs with
	a few thousand edges.  In order to handle larger graphs, we drop the condition
	that each vertex $v$ has to be moved to its crossing-minimal position and
	compute a position that is only optimal with respect to a small random subset
	of the edges. In our theoretical contribution, we consider drawings that
	contain for each edge $uv \in E$  and each position $p \in \R^2$ for $v$
	$o(|E|)$ crossings.  In this case, we prove that with a random subset of the
	edges of size $\Theta(k \log k)$ the \emph{co-crossing} number of a degree-$k$
	vertex $v$, i.e., the number of edge pairs $uv \in E, e \in E$ that do
	\emph{not} cross, can be approximated by an arbitrary but fixed factor
	$\delta$ with high probability. In our experimental evaluation, we show that
	the randomized approach reduces the number of crossings in graphs with up to
	$13\,000$ edges considerably. The evaluation suggests that depending on the
	degree-distribution different strategies result in the fewest number of
	crossings.	
\end{abstract}


\section{Introduction}
The minimization of crossings in geometric drawings of graphs is a fundamental
graph drawing problem.  In general the problem is
$\cNP$-hard~\cite{garey1983crossing,bienstock1991some} and has been studied from
numerous theoretical perspectives; see Vrt'o~\cite{crBibVrto}.  Until recently
only the topological setting, where edges are drawn as topological curves, has
been considered in
practice~\cite{DBLP:journals/algorithmica/GutwengerMW05,bcg-hgdv-13,chimani_et_al:LIPIcs:2016:5922}.
In our previous paper~\cite{doi:10.1137/1.9781611975055.12} we describe
geometric heuristics that compute straight-line drawings of graphs with
significantly fewer crossings compared to common energy-based layouts. One of
the heuristics is the \emph{vertex-movement approach} that iteratively moves a
single vertex $v$ to its \emph{crossing-minimal position}, i.e., a position
$p^\star$ so that crossings of edges incident to $v$ are minimized.
Unfortunately, the worst-case running time to compute this position is
super-quadratic in the size of the graph as the following theorem states.

\begin{theorem}[Radermacher et al.~\cite{doi:10.1137/1.9781611975055.12}]
	\label{theorem:opt:placement}
	The crossing-minimal position of a degree-$k$ vertex $v$ with respect to a
	straight-line drawing $\Gamma$ of a graph $G=(V,E)$ can be computed in
	\RCMRunningTime time, where $n = |V|, m=|E|$.
\end{theorem}

This is not only a theoretical upper bound on the running time but is also a
limitation that has been observed in practice.
The implementation we used previously requires considerable time to compute
drawings with few crossings. For this reason we were only able evaluate our
approach on graphs with at most $200$ edges.  For example, on a class of graphs
that have $64$ vertices and $196$ edges our implementation already required on
average about 35 seconds to compute a drawing with few crossings.

Energy-based methods are common and well engineered tools to draw
graphs~\cite{DBLP:journals/corr/abs-1201-3011}. For example, the aim of
\emph{Stress Majorization} (or simply \stress) is to compute a drawing such that
the Euclidean distance of each two vertices corresponds to their
graph-theoretical distance~\cite{Gansner2005}.  The algorithm has been
engineered to handle graphs with up to $10^6$ vertices and $3\cdot 10^6$
edges~\cite{7889042}.  Kobourov~\cite{DBLP:journals/corr/abs-1201-3011} claimed
that \stress tends to minimize the number of crossings. In our previous
experimental evaluation~\cite{doi:10.1137/1.9781611975055.12} we demonstrated
that the statement is not true for a varied set of graph classes.

Fabila-Monroy and L\'opez~\cite{JGAA-328} introduced a randomized algorithm to
compute a drawing of $K_n$ with a small number of crossings. Many best known
upper bounds on the rectilinear crossing number of $K_n$, for $44 \leq k \leq
99$, are due to this approach~\cite{RCNProject}. The algorithm iteratively
updates a set $P$ of $n$ points, by replacing a random point $p \in P$ by a
random point $q$ that is close to $p$, if $q$ improves the number of crossings.
Since the number of crossings of $K_n$ is in $\Theta(n^4)$, the bottleneck
of their approach is the running time for counting the number of crossings
induced by $P$. 
A similar randomized approach has been used to maximize the smallest
crossing angle in a straight-line
drawing~\cite{DBLP:conf/gd/BekosFGH0SS18,DBLP:conf/gd/DemelDMRW18}. The approach
iteratively moves vertices to the best position within a random point set. 

\subparagraph*{Contribution.} The main contribution of this paper is to engineer
the \emph{vertex-movement approach} for the minimization of crossings in
geometric drawings described in \cite{doi:10.1137/1.9781611975055.12} to be
applicable on graphs with a few thousands vertices and edges. 
\begin{enumerate}
	\item In Section~\ref{sec:impl} we introduce so-called \emph{bloated duals of
		line arrangements}, a combinatorial technique to construct a dual
		representation of general line arrangements. In our application this results
		in an overall speed-up of about a factor of $20$ in comparison to the recent
		implementation. This speed-up is necessary but not sufficient to handle
		graphs with a few thousands vertices and edges.
	\item In Section~\ref{sec:random_sampling} we demonstrate that taking a small
		random subset of the edges is sufficient to compute drawings with few
		crossings.  Moreover, in Section~\ref{sec:apx_cocrossings} we prove that
		under certain conditions the randomized approach is an approximation of the
		co-crossing number of a vertex, with high probability.
	\item Based on the insights of the evaluation in
		Section~\ref{sec:eval_restricted}, we introduce a weighted sampling
		approach. A comparison to a restrictive approach of sampling points suggests
		that the degree-distribution of the graph is a good indicator to decide
		which approach results in fewer crossings.
	\item Overall, our experimental evaluation shows that we are now able to
		handle graphs with $12\,000$ edges, which are $60$ times more than the
		graphs that have been considered in the evaluation in
		\cite{doi:10.1137/1.9781611975055.12}.
\end{enumerate}

\section{Preliminaries}
\label{sec:prelim}

\begin{figure}
	\centering
	\includegraphics{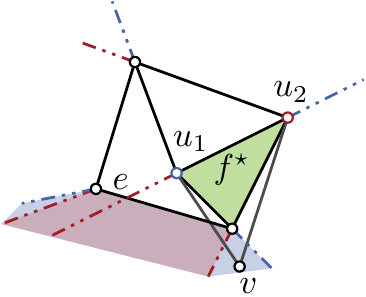}
	\caption{The black, blue and red segments show the arrangement
	$\Arr(\Gamma,v)$ of the black drawing $\Gamma$. The blue and red region show
the complement of the visibility regions of $u_1$ and $u_2$, respectively, and
the edge~$e$. The green region is crossing minimal.}
\label{fig:crossing_min}
\end{figure}

We repeat some notation from~\cite{doi:10.1137/1.9781611975055.12}.  Let
$\Gamma$ be a straight-line drawing of a $G=(V, E)$.  Denote by $N(v) \subseteq
V$ the set of neighbors of $v$ and by $E(v) \subseteq E$ the set of edges
incident to $v$.  For a vertex $v \in V$, denote by $\moveVertex{v}{p}$ the
drawing that is obtained from $\Gamma$ by moving the vertex $v$ to the point
$p$. We denote the number of crossings in a drawing $\Gamma$ by $\Cr(\Gamma)$,
the number of crossings on edges incident to $v$ by $\Cr(\Gamma, v)$, and we
refer with $\Cr(\Gamma, e, f)$ to the number of crossings on two edges $e$ and
$f$ in $\Gamma$, i.e., $\Cr(\Gamma, e, f) \in \{0,1\}$ if $e\not=f$.
For a point $u$ and a segment $e$, denote by $\Vr(u, e)$ the \emph{visibility
region of $u$ and $e$}, i.e., the set of points $p \in \R^2$ such that the
segment $up$ and $e$ do not intersect. Moreover, let $\Bd(u, e)$ be the boundary
of $\Vr(u,e)$.
Let $\Arr(\Gamma, v)$ be the arrangement over all boundaries $\Bd(u, e)$ for
each neighbor $u \in N(v)$ of $v$ and each edge $e \in E \setminus E(u)$; see
\cref{fig:crossing_min}. The arrangement has the property that two points
$p$ and $q$ in a common cell of $\Arr(\Gamma, v)$ induce the same number of
crossings for $v$, i.e., $\Cr(\moveVertex{v}{p}, v)= \Cr(\moveVertex{v}{q},
v)$~\cite{doi:10.1137/1.9781611975055.12}.  Thus, the computation of a crossing
minimal position $p^\star$ reduces to finding a \emph{crossing-minimal region}
$f^\star$ in $\Arr(\Gamma,v)$.

For our experiments, we used two different compute servers. Both
systems ran with an openSUSE Leap~15.0 operating system.  All algorithms were
compiled with \texttt{g++} version 7.3.1 with optimization mode \texttt{-O3}.
System~1 was used for  running time experiments, i.e., for the experiments
evaluated in \cref{sec:bd_runtime} and in \cref{sec:eval_restricted}.
System~2 is used for the experiments evaluated in \cref{sec:eval_weighted}.
\begin{description}
	\item[System 1] Intel Xeon(tm) E5-1630v3 processor clocked at 3.7~GHz, 128\,GB
		RAM.
	\item[System 2] Two Intel Xeon(tm) E5-2670 CPU processors clocked at
		2.6~GHz, 64\,GB RAM.  \end{description}

\section{Efficient Implementation of the Crossing-Minimal Position}
\label{sec:impl}

The vertex-movement approach iteratively moves a single vertex to its
crossing-minimal position. The running time of the overall algorithm crucially
depends on an efficient computation of this operation.
Therefore the aim of this section is to provide an efficient implementation of
the crossing-minimal position of a vertex. Our previous
implementation~\cite{doi:10.1137/1.9781611975055.12} heavily relies on
CGAL~\cite{cgal:eb-17a}, which follows an exact computations paradigm and uses
exact number types to, e.g., represent coordinates and intermediate results.
This helps to ensure correctness but considerably increases the running time of
the algorithms. We introduce an approach to compute the crossing-minimal
position that drastically reduces the usage of exact computations.

Computing a crossing-minimal position of a vertex $v$ is equivalent to computing
a crossing-minimal region $f^\star$ in the arrangement $\Arr(\Gamma, v)$.  The
region $f^\star$ of a vertex $v$ can be computed by a breadth-first search in
the dual graph $\Arr(\Gamma, v)^\star$. Thus, the time-consuming steps to
compute $f^\star$ are to construct the arrangement $\Arr(\Gamma,v)$ and then to
build the dual $\Arr(\Gamma, v)^\star$.  Instead of computing the dual
$\Arr(\Gamma,v)^\star$ we construct a so-called \emph{bloated dual}
$\Arr(\Gamma, v)^+$. The advantage of this approach is that it suffices to
compute the set of intersecting segments in $\Arr(\Gamma, v)$ to construct
$\Arr(\Gamma,v)^+$ and it is not necessary to compute the arrangement
$\Arr(\Gamma, v)$ itself, i.e., the exact coordinates of each intersection.

\begin{figure}[tb]
	\centering
	\begin{subfigure}[b]{.3\textwidth}
		\includegraphics[page=1]{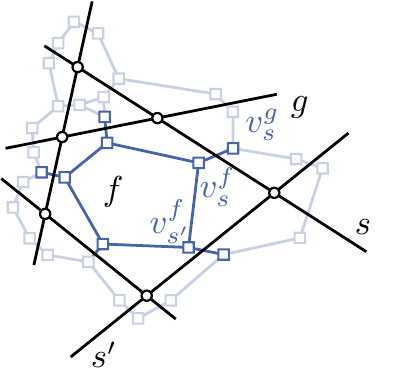}	
		\caption{}
		\label{fig:bloated_dual:def}
	\end{subfigure}
	\begin{subfigure}[b]{.3\textwidth}
		\includegraphics[page=2]{fig/bloated_dua.pdf}	
		\caption{}
		\label{fig:bloated_dual:segment_edges}
	\end{subfigure}
	\begin{subfigure}[b]{.3\textwidth}
		\includegraphics[page=3]{fig/bloated_dua.pdf}	
		\caption{}
		\label{fig:bloated_dual:face_edges}
	\end{subfigure}

	\caption{(a)~Bloated dual $\Arr^+$ (blue) of an arrangement $\Arr$ (black).
		Inserting edges dual to a segment $s$ (b) and within a face (c).}
\end{figure}

Let $S$ be a set of line segments and let $\Arr$ be the arrangement of $S$.  A
\emph{bloated dual of $\Arr$} is a graph $\Arr^+$ that has the following
properties (compare \cref{fig:bloated_dual:def}):
\begin{enumerate}[(i)] \item each edge $e$ incident to a face $f$ corresponds to a
		vertex $v_e^f$ in $\Arr^+$,
	\item if two distinct segments $s, s' \in S$ of $f$ have a common intersection on the
		boundary of $f$, then $v_{s}^fv_{s'}^f \in E(\Arr^+)$, and
	\item for two distinct faces $f,g$ sharing a common segment~$s$, there is an edge
		$v_s^fv_s^g \in E(\Arr^+)$.	
\end{enumerate} 
Note that given a crossing-minimal face and $v^f_{s_0}$, the geometric
representation of $f$ has to be computed in order to compute a crossing-minimal
position $p \in f$.  Further a vertex $v^f_{s_0}$ belongs to a cycle
$v^f_{s_0}, v^f_{s_1}, \dots v^f_{s_k}$. Then, the geometric representation of
the boundary of $f$ can be computed by intersecting the segments $s_{i}$ and
$s_{i+1}$, where we set $k+1=0$. In the following, we will show that it is sufficient
to know the order in which the segments in $S$ intersect to construct the
bloated dual. Thus, exact number types only have to be used to determine the
order of two segments whose intersections with a third segment $s$ have a small
distance on $s$.

We construct the bloated dual of $\Arr$ in two steps. First, we insert all
vertices $v_s^f, v^g_{s}$ and the corresponding edge $v_s^fv_s^g$. In the second
step, we insert the remaining edges $v_s^fv_{s'}^f$ within a face $f$. For a
compact description we assume that no intersection point of two segments is an
endpoint of a segment. We define the \emph{source of $s$}  and \emph{target of
$s$} to be the lexicographically smallest and largest point on $s$,
respectively. We direct each segment $s$ from its source to its target.

Let $p_1, p_2, \dots, p_l$ be the intersection points on a segment $s$ in
lexicographical order. These intersection points correspond to a set of left
faces $f^L_1, f^L_2, \dots, f^L_{l+1}$ and to a set of right faces $f^R_1,
f^R_2, \dots, f^R_{l+1}$, such that $f^L_i$ and $f^R_i$ share parts of their
boundary; see \cref{fig:bloated_dual:segment_edges}. Thus, we can associate
a set of vertices $v^L_i, v^R_i, 2 \leq i \leq l+1$, with $s$, and add the edges
$v^L_{i}v^R_i$ to $\Arr^+$. Note that only the order and not the actual
coordinates of the points $p_1, \dots, p_l$ has to be known to insert the edges.
Thus, given the set of segments that intersect $s$, an exact number type is only
necessary to determine the order of two segments $s_i$ and $s_j$ whose
intersection points $p_i$ and $p_j$ on $s$ have a small distance.

We now add the remaining edges within a face $f$. Let $S' =
\{s_1, \dots, s_k\} \subseteq S$ be the set of segments that intersect $s$ in
$p_i$; see \cref{fig:bloated_dual:face_edges}. The two segments $s^L,s^R \in
S'$ that lie on the boundary of $f^L_i$ and $f^R_{i}$ can be determined as
follows. To find the segment $s^L$, we distinguish two cases.  First, assume
that there exists a segment $s' \in S'$ whose source is left of $s$. Observe
that if there is a segment $s''$ whose target is left of $s$, the segment $s''$
cannot be the segment $s^L$.  Thus, we assume without loss of generality that
all sources of segments in $S^i_s$ are left of $s$. Then a segment $s' \in S'$
is the segment $s^L$ if and only if the segment $s'$ and each segment $s'' \in
S' \setminus \{s'\}$ form a right turn.  Now consider the case that there is no
segment whose source is left of $s$. Then a segment $s'$ is $s^L$ if and only if
the segment $s'$ and each segment $s'' \in S' \setminus \{s'\}$ form a left
turn. The segment $s^R$ can be determined analogously.

\subparagraph*{Implementation Details.}

We give some implementation details which allow us to efficiently implement the
construction of the bloated dual.  We use the index of a vertex to decide
whether it is left or right of $s$, i.e., vertices with an odd index are left of
$s$ and vertices with an even index are right of $s$.  The fact that each vertex
of $\Arr^+$ has degree at most 3 allows us to represent $\Arr^+$ as a
single array $B$ of size $3n$, where $n$ is the number of vertices of $\Arr^+$.
The vertices incident to a vertex $v_i$ occupy the cells $B[3i], B[3i+1]$ and
$B[3i+2]$. Moreover, each pair of segments in $S$ can be handled independently
to construct the bloated dual. This enables a parallelization over the segments
in $S$.

\subsection{Evaluation of the Running Time}
\label{sec:bd_runtime}

\begin{figure}
	\begin{subfigure}[b]{.24\textwidth}
		\includegraphics[page=1]{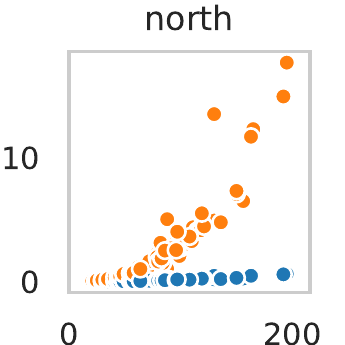}
		\caption{}
	\end{subfigure}
	\begin{subfigure}[b]{.24\textwidth}
		\includegraphics[page=2]{plots/runtime/runtime.pdf}
		\caption{}
	\end{subfigure}
	\begin{subfigure}[b]{.24\textwidth}
		\includegraphics[page=3]{plots/runtime/runtime.pdf}
		\caption{}
	\end{subfigure}
	\begin{subfigure}[b]{.24\textwidth}
		\includegraphics[page=4]{plots/runtime/runtime.pdf}
		\caption{}
	\end{subfigure}

	\caption{Comparing the running time of two approaches (orange \textsc{Precise},
	blue \textsc{Bd}) to compute the crossing
	minimal region. Each point corresponds to a graph $G$. The $x$-axis shows the
	number of edges of $G$. The $y$-axis depicts the running time in seconds to compute the crossing
	minimal regions for all vertices of $G$.} 
	\label{fig:runtime_comparison}
\end{figure}

In this section, we compare the running time of the two approaches to compute
the crossing-minimal region of a vertex. We refer with \textsc{Precise} to the
approach that uses CGAL to compute the crossing minimal region and with
\textsc{Bd} to the approach based on the bloated dual.
In order to compute all intersecting segments, we use a naive implementation of
a sweep-line algorithm~\cite{1675432}. In this approach all segments within a
specific interval are pairwise checked for an intersection. This has the
advantage that the computation is independent of the coordinates of the
intersection.

The experimental setup is as follows. Given a drawing $\Gamma$ of a graph $G$,
we are interested in the running time of moving all vertices of a graph to their
crossing-minimal positions. Therefore, we measure the running time of computing
the crossing-minimal regions of all vertices. In order to guarantee the
comparability of the two approaches, we use the same vertex order and only
compute the crossing-minimal region but do not update the positions of the
vertices.
We use the same set of benchmark graphs used in~\cite{doi:10.1137/1.9781611975055.12}:
\textsc{North}\footnote{\label{fn:north_rome}http://graphdrawing.org/data.html},
\textsc{Rome}\textsuperscript{\ref{fn:north_rome}},
graphs that have \textsc{Community} structure, and
\textsc{Triangulations} on 64 vertices with an additional $10$ random edges.
For each graph class, 100 graphs were selected uniformly at random. We use the
implementation of \stress~\cite{Gansner2005} provided by
\ogdf~\cite{ogdf} (snapshot 2017-07-23) to compute an initial layout of the
graphs.
 

The plots in \cref{fig:runtime_comparison} shows the results of the
experiments. Each point in the plot corresponds to the running time of
computing all crossing-minimal region of a single graph. The plot shows that the
\textsc{Bd} implementation is considerably faster than the \textsc{Precise}
implementation. For each graph class, we achieve on average a speed-up of at
least $20$. The minimum speed-up on the \textsc{North} graphs is $8$. For each graph
class, the speed-up is at least 18 for at least 75 out of 100 instances. 

\section{Random Sampling}
\label{sec:random_sampling}

\begin{figure}
		\centering
		\includegraphics{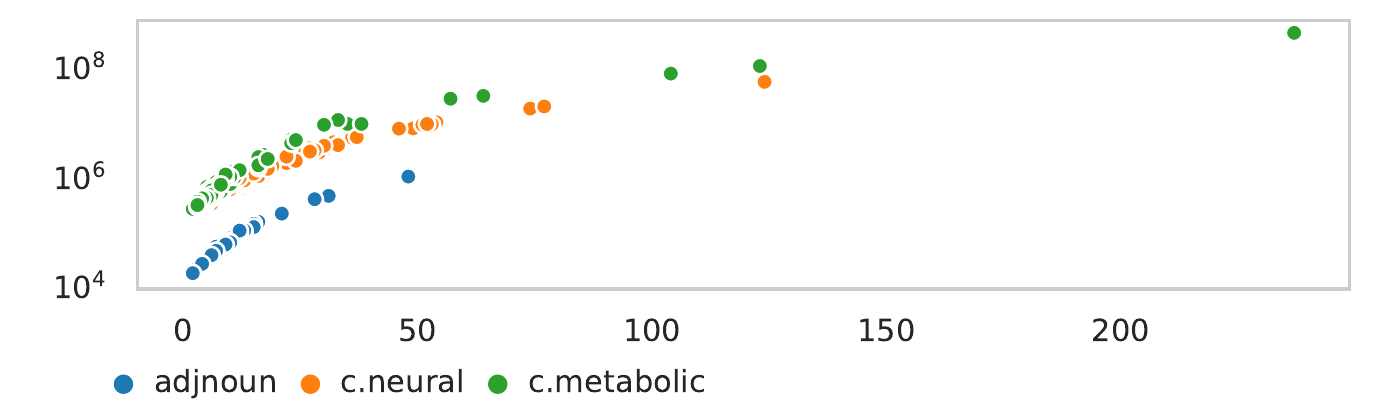}
		\caption{The $x$-axis shows the vertex-degree and the $y$-axis the number of
		intersecting edges in the arrangement $\Arr(\Gamma, v)$. The $y$-axis is in
	$\log$-scale.}
	\label{fig:arr_size}
\end{figure}

The worst-case running time of computing the crossing-minimal region of a vertex $v$
is super-quadratic in the size of the graph, see
\cref{theorem:opt:placement}.  \cref{fig:arr_size} shows the number
of intersecting segment in the arrangement $\Arr(\Gamma, v)$ compared to the
vertex-degree of $v$, for vertices of three selected graphs with at most 2\,133
edges, compare \cref{tab:stats:dimacs:512}. For these graphs the arrangement
already contains up to $440\,685\,519$ intersecting segments. Indeed, we were not
able to compute the number of intersections for all vertices of the graph
c.metabolic, since the algorithm ran out of memory first. Due to the high number
of intersections in graphs with a high number of edges or a large maximum
vertex-degree, it is for these graphs  infeasible to compute a crossing-minimal
position of a vertex. This motivates the following question: Is a small subgraph
of $G$ sufficient to considerably reduce the number of crossings in a given
drawing?

To address this question, we follow the \emph{vertex-movement approach}. Let
$\Gamma_0$ be a drawing of $G$ and let $v_1, v_2, \dots, v_n$ be an ordered
set of the vertices $V$ of $G$. For each vertex $v_i$ we obtain a new drawing
$\Gamma_i$ from the drawing $\Gamma_{i-1}$ by moving $v_i$ to a new position
$p^\star_i$. To compute the new position we consider a \emph{primal} sampling
approach, i.e., a sampling of points in the solution space $\R^2$, and a
\emph{dual} sampling approach, i.e., a sampling of edges that induce constraints
to the solution space.

More formally, we consider the following approach to compute a new position of a
single vertex $v_i$. Let $S_i \subset E$ be a uniform random subset of the edges
of $G$ and let $V(S_i) \subset V$ be the vertices that are incident to an edge
in $S_i$.  The graph $G|_{S_i}= (V(S_i) \cup N(v_i) \cup \{v_i\}, S_i \cup
E(v_i))$ induces a drawing $\Gamma|_{S_i}$ in $\Gamma_{i-1}$. Let $R_i$ be the
crossing-minimal region of $v_i$ with respect to the drawing $\Gamma|_{S_i}$.
Recall that for $S_i = E$ the region $R_i$ has the property that
$\Cr(\moveVertex[\Gamma|_{S_i}]{v_i}{p}, v_i) =
\Cr(\moveVertex[\Gamma|_{S_i}]{v_i}{q}, v_i)$ for any two points $p,q \in R_i$,
compare \cref{sec:prelim}. If $S_i$ is a strict subset of $E$, then $R_i$
does not necessarily have this property anymore. For this reason, let $P_i
\subset R_i$ be a set of uniform random points and let $p_i^\star \in P_i \cup
\{p_i'\}$ be the point that minimizes $\Cr(\moveVertex{v}{p_i^\star},v_i)$,
where $p_i'$ is the position of $v_i$ in $\Gamma_{i-1}$.

This remainder of this section is organized as follows. First, we analyze the dual
sampling from a theoretical perspective (\cref{sec:apx_cocrossings}),
followed by an experimental evaluation that compares the primal to the dual
sampling (\cref{sec:eval_restricted}). Finally, based on the insights
from this evaluation, we introduce in \cref{sec:eval_weighted} a
\emph{weighted} sampling approach that is less restrictive than the dual
sampling.

\subsection{Approximating the Co-Crossing Number of a Vertex}
\label{sec:apx_cocrossings}

In this section we study the dual sampling approach, i.e., the sampling of
edges, with tools introduced in the context of the theory of VC-dimension.  A
thorough introduction into the theory of VC-dimension can be found in
Matou{\v{s}}ek's \emph{Lectures on Discrete
Geometry}~\cite{matouvsek2002lectures}.
For a fixed vertex $v$, a drawing~$\Gamma$ is \emph{$\varepsilon$-well behaved}
if for each point $p \in \R^2$ and each vertex $u\in N(v)$, the edge $uv$
crosses at most $(1-\varepsilon)|E|$ edges in the drawing~$\moveVertex{v}{p}$.
The \emph{co-crossing number} $\coCr(\Gamma, v)$ of a vertex~$v$ is the number
of edge pairs $e \in E$ and $uv \in E$ that do not cross. 
We show that given an $\varepsilon$-well-behaved drawing $\Gamma$ of a graph
$G=(V, E)$ and a degree-$k$ vertex $v$, a random sample $S \subset E$ of size
$\Theta(k \log k)$ enables us to compute a position $q^\star$ whose co-crossing
number is a $(1-\delta)$-approximation of the co-crossing number of a
vertex~$v$.
Note that we are not able to guarantee that a large co-crossing number of a
vertex $v$ implies a small crossing number of $v$.
On the other hand, the co-crossing number is of interest for a variety of
(sparse) graph. For example, drawings that contain many triangles are
$\varepsilon$-well-behaved, since every line intersects at most two segments of
a triangle.

A \emph{set system} is a  tuple $(X,
\mathcal F)$ with a base set $X$ and $\mathcal F \subseteq 2^X$. In the
following, we assume $X$ to be finite.  For some parameters $\varepsilon, \delta
\in (0,1]$, a set $S \subseteq X$ is a \emph{relative $(\varepsilon,
\delta)$-approximation} for the set system $(X, \mathcal F)$ if for each $R \in
\mathcal F$ the following inequality holds. 

\begin{equation}\label{eq:ed_apx}
	\size{ \frac{\size{S \cap R}}{ \size{S}} - \frac{\size{R}}{\size{X}}}
	\leq \delta \max\{ \frac{\size{R}}{\size{X}}, \varepsilon\}
\end{equation}

The proof of the following proposition and of proofs of statements that are
marked with ($\star$) can be found in
\cref{appendix:missing_proofs}.

\begin{restatable}[$\star$]{proposition}{obsrelativeapx}
	\label{obs:relative_apx}
	For $\varepsilon, \delta \in (0, 1]$, let $S$ be an $(\varepsilon,
	\delta)$-approximation of the set system $(X, \mathcal F)$.
	If every $R \in \mathcal F$ has size at least $\varepsilon
	|X|$ then \cref{eq:ed_apx} can be rewritten as follows:
	\[
	(1-\delta) \size{R} \leq |X| \frac{\size{S \cap R}}{ \size{S}} \leq
	(1+\delta) \size{R} .\]
\end{restatable}

Let $\mathcal F|_A  = \{ R \cap A \mid R \in \mathcal F \}$ be the restriction of
$\mathcal F$ to a set $A \subseteq X$.  A set $A \subseteq X$ is \emph{shattered by}
$\mathcal F$ if every subset of $A$ can be obtained by an intersection of
$A$ with a set $R \in \mathcal F$, i.e., $\mathcal F|_A=
2^A$. The \emph{VC-dimension} of a set
system~$(X, \mathcal F)$ is the size of the largest subset $A \subseteq X$ such
that $A$ is shattered by $\mathcal F$~\cite{Vapnik71}.



\begin{theorem}[Har-Peled and Sharir~\cite{HarPeled2011}, Li et al.~\cite{LI20015}]
	\label{theorem:rel_eps_approx}
	Let $(X, \mathcal F)$ be a finite set system with VC-dimension $d$, and
	let $\delta, \varepsilon, \gamma \in (0,1]$. A uniform random
	sample $S \subseteq X$ of size 
		\[ \Theta \left( \frac{ d \cdot \log \varepsilon^{-1} + \log \gamma^{-1}} { \varepsilon \delta ^2}\right)\]
	is a relative $(\varepsilon,
	\delta)$-approximation for $(X, \mathcal F)$ with probability $(1-\gamma)$.
\end{theorem}

For a vertex $u \in N(v)$, let $\coCrEdge{uv}{\Gamma} = \{ e \in E \mid
\Cr(\Gamma, e, uv) = 0 \}$ denote the set of edges that are not crossed by the
edge $uv$ in $\Gamma$.  Then we have $\coCr(\Gamma, v) = \sum_{u \in N(v)}
\size{\coCrEdge{uv}{\Gamma}}$.  Moreover, let $\coCrEdge{uv}{p} =
\coCrEdge{uv}{\moveVertex{v}{p}}$. 
Then the set $\co{\mathcal F_{uv}} = \bigcup_{p \in \R^2}
\left\{\coCrEdge{uv}{p}\right\}$ contains for each drawing $\moveVertex{v}{p}$
the set of edges that are not crossed by the edges $uv$, i.e,
$\coCrEdge{uv}{p}$.
%
In particular $(E, \co{\mathcal F_{uv}})$ is a set system and we will prove that
it has bounded VC-dimension. This allows us to approximate the number of edges
that are not crossed by the edge $uv$. We facilitate this to approximate the
co-crossing number of a vertex for $\varepsilon$-well behaved drawings.

\begin{lemma}
	\label{lemma:vc_dimension}
	The VC-dimension of the set system $(E, \co{\mathcal F_{uv}})$ is at most 8.
\end{lemma}

\begin{proof} 
	Recall that that vertex $u$ has a fixed position. Let  $\Bd(u, e)$ be the
	boundary of the visibility region of $u$ and the edge $e \in E$.  Let $\Arr$
	denote the arrangement of all boundaries $\Bd(u,e), e \in E$.  Let $F$ be the
	set of faces in $\Arr$.  Note that by Lemma~3.1
	in~\cite{doi:10.1137/1.9781611975055.12} for every two points $p,q \in f$ the
	sets $E_p$ and $E_q$ of edges that have a non-empty intersection with the edge
	$uv$ when $v$ is moved to $p$ and $q$, respectively, coincide. Hence, the set
	$E_f \subseteq E$ of edges that cross the edge $uv$, in the drawing obtained
	from $\Gamma$ where $v$ is moved to an arbitrary position in $f$, is well
	defined.  Thus, the number of faces $|F|$ is an upper bound for
	$\size{\co{\mathcal F_{uv}|_A}}$ for every $A \subset E$. Note that there may
	be subsets of $E$ that are represented by more than one face. Moreover,
	observe that the visibility region $\Vr(u, e)$ is the intersection of three
	half-planes. Let $l^1_e, l^2_e,l^3_e$ be the supporting lines of these
	half-planes and let $\Arr'$ be the arrangement of lines $l^i_e, e\in E$.
	Hence, the number of faces in the arrangement $\Arr'$ of $3m$ lines is an
	upper bound for $|F|$, with $m=|E|$.  The number of faces $|F'|$ of $\Arr'$ is
	bounded by $f(m) := 3m(3m- 1) / 2 + 1$~\cite{10.2307/2686448}. Thus, it is not
	possible to shatter a set $A \subset E$ if the number of faces $|F'|$ is
	smaller than the number of subsets of $A$. The largest number for which the
	equality $2^m \leq f(m)$ holds is between $8$ and $9$. Since $2^m$ grows
	faster than $f(m)$, the largest set that can possibly be shattered has size at
	most $8$. 
\end{proof}


Due to \cref{obs:relative_apx} and \cref{theorem:rel_eps_approx} a relative
$(\varepsilon, \delta)$-approximation $S_u$ of $(E, \co{\mathcal F_{uv}})$
allows us to approximate the number of edges that are not crossed by the edge
$uv$.
In the following we show that we can approximate the co-crossing number of a
vertex $v$ in any drawing $\moveVertex{v}{p}$ if we are given a relative
$(\varepsilon, \delta)$-approximation $S_u$ for each vertex $u$ that is adjacent
to $v$. 
The number $|\coCrEdge{uv}{p} \cap S_u| / |S_u| $ corresponds to the relative
number of edges in $S_u$ that  are not crossed by the edge  $uv$. Hence, the
function $\lambda(p) = |E| \sum_{u \in U} {\tsize{\coCrEdge{uv}{p}\cap
S_u}}/{\tsize{S_u}}$ can be seen as an estimation of $\coCr(p) =
\coCr(\moveVertex{v}{p}, v)$.
%

\begin{restatable}[$\star$]{lemma}{lemmaapxsum}
	%
	\label{lemma:apx:sum}
	Let $\varepsilon, \delta \in (0,1]$ be two parameters and let $\Gamma$ be an
	$\varepsilon$-well behaved drawing of~$G$.  For every $u \in N(v)$, let $S_u$
	be a relative $(\varepsilon, \delta)$-approximation of the set system $(E,
	\co{\mathcal F_{uv}})$.  Then $(1-\delta) \coCr(p) \leq \lambda(p)  \leq
	(1+\delta) \coCr(p)$ holds for all $p \in \R^2$.
	%
\end{restatable}

Assume that $\varepsilon, \delta, \gamma \in (0,1) $ are constants.
\cref{lemma:apx:sum} shows that $k$ independent samples $S_u$ of constant size
approximate the co-crossing number of $v$. By slightly increasing the number of
samples, we can use a single set $S$ for all neighbors $u$.  This reduces the
running time from $O(k^3\log k)$ to $O(k^2\log^3 k)$.

\begin{restatable}[$\star$]{lemma}{lemmarandomsample}
	\label{lemma:random_sample}
	Let $v$ be a degree-$k$ vertex and let $\varepsilon, \delta, \gamma \in
	(0,1]$ with $\gamma \leq 1 / k$.
	A  uniformly random sample $S \subseteq E$ of size 
	$\Theta \left( (\log \varepsilon^{-1} + \log \gamma^{-1}) / (\varepsilon
	\delta ^2)\right)$ 
	is a relative $(\varepsilon, \delta)$-approximation the set system $(E,
	\co{\mathcal F_{uv}})$ with probability $1 - k\gamma$, for each $uv \in E$.
	%
\end{restatable}

With \cref{lemma:apx:sum} and \cref{lemma:random_sample} at hand, we
have all the necessary tools to prove the main theorem.

\begin{theorem}
	\label{theorem:apx:general}
	Let $\varepsilon, \delta,
	\gamma \in (0,1]$ be three constants and let $G=(V,E)$ be a graph with a
	$\varepsilon$-well behaved drawing $\Gamma$ and let $v \in V$ be a degree-$k$
	vertex. Let $p^\star$ be the position that maximizes $\coCr(\moveVertex{v}{p^\star},
	v)$.  A $(1-\delta)$-approximation of $\coCr(\moveVertex{v}{p^\star})$ can be
	computed in $O\left ( k^2 \log^3 k \right)$ time with probability $1 - \gamma$.
\end{theorem}

\begin{proof} 
	Let $\gamma'= \gamma \cdot k^{-1}$ and $\delta' = \delta / 2$.
	Let $S\subseteq E$ be a uniformly random sample of size $\Theta \left( (\log
	\varepsilon^{-1} + \log \gamma'^{-1}) / (\varepsilon \delta'^2)\right)$.
	According to \cref{lemma:random_sample}, for each $uv \in E$, the sample
	$S$ is a $(\varepsilon, \delta')$-approximation of the $(E, \co{\mathcal
	F_{uv}})$ with probability $1 - k \gamma' = 1 - \gamma$. 

	According to \cref{lemma:apx:sum} the expected number of crossing-free
	edges $\lambda(p)$ is a $(1-\delta)$-approximation of $\coCr(p)$, i.e.,
	$(1+\delta') \coCr(q) \geq \lambda(q) \geq (1-\delta') \coCr(q)$.  Let $p^\star$ be the
	position that maximizes $\coCr(p)$ and let $q^\star$ be the position that
	maximizes $\lambda(q)$. Hence, we have $\lambda(q^\star) \geq
	\lambda(p^\star)$. Observe that over $\delta'>0$ the inequality $(1-\delta')/
	(1+\delta') \geq 1 - 2\delta'$ holds.  We use this to prove that $\coCr(q^\star)
	\geq (1-2\delta') \coCr(p^\star)$.
	
	\begin{align*}
		\coCr(q^\star)
		\geq \frac{1}{(1+\delta')} \lambda(q^\star) 
		\geq \frac{1}{(1+\delta')} \lambda(p^\star)
		\geq \frac{1-\delta'}{1+\delta'}  \coCr(p^\star)  
		\geq (1-2\delta')\coCr(p^\star)
	\end{align*}
	Plugging in the value $\delta / 2$ for $\delta'$ yields that
	$\coCr(q^\star)$ is a $\delta$-approximation of $\coCr(p^\star)$.
	%
	Since the three parameters $\varepsilon, \delta, \gamma$ are
	constants, the size of the sample $S$ is in $\Theta( \log k)$.  
	Recall that the running time to compute the crossing-minimal position of $v$
	in a drawing $\Gamma$ is \RCMRunningTime
	(\cref{theorem:opt:placement}). Thus the position $q^\star$ can be
	computed in $O(k \log k + \log k)^2  \log (k \log k + \log k))$ time,
	since $m = |S| \in  \Theta(\log k)$ and $n \leq 2m$. The following estimation
	concludes the proof.
	\[
		O \left( k^2 \log^2  k \log (k \log k)\right) 
		= O \left( k^2 \log^2  k \log (k^2)\right) 
		= O(k^2 \log^3 k)
	\]
	%
\end{proof}



Note that the previous techniques can be used to design a $\delta$-approximation
algorithm for the crossing number of a vertex. But this requires drawings of
graphs where at least $\varepsilon |E|$ edges, i.e., $\Omega(|E|)$, are crossed.
This restriction is not too surprising, since sampling the set of edges can
result in an arbitrarily bad approximation for a vertex whose crossing-minimal
position induces no crossings.

\subsection{Experimental Evaluation}
\label{sec:eval_restricted}

In this section we complement the theoretical analyses of the random sampling
of edges with an experimental evaluation. We first introduce our benchmark
instances, followed by a  description of a preprocessing step to reduce trivial
cases and a set of configurations that we evaluate.

\subparagraph{Benchmark Instances.}
We evaluate our algorithm on graphs from three different
sources.
\begin{description}
	\item[DIMACS] The graphs from this classes are selected from the 10th \dimacs
		Implementation Challenge - Graph Partitioning and Graph
		Clustering~\cite{DBLP:reference/snam/BaderKM00W18}.
	\item[Sparse MC] Inspired by the selection of benchmark graphs
		in~\cite{7889042}, we selected a few arbitrary graphs from the Suite Sparse
		Matrix Collection (formerly known as the Florida Sparse Matrix
		Collection)~\cite{Davis:2011:UFS:2049662.2049663}. 
	\item[$k$-regular] For each $k=3,6,9$ we computed $25$ random $k$-regular
		graphs on $1000$ vertices following the model of Steger and
		Wormald~\cite{steger_wormald_1999}.  \end{description}

\subparagraph{Preprocessing.} Some of the benchmark graphs contain multiple
connected components.  Moreover, we observed that the \stress layout introduces
crossings with edges that are incident to a degree-1 vertex. In both cases,
these  crossings can be removed. Therefore, we reduce the benchmark instances so
that they do not contain these trivial cases as follows.  First, we evaluate
only the connected component $G_C$ of each graph $G$ that has the highest number
of vertices. Further, we iteratively remove all vertices of degree $1$ from
$G_C$.  

The vertex-movement approach takes an initial drawing of a graph as input. Note
that the experimental results in~\cite{doi:10.1137/1.9781611975055.12} showed
that drawings obtained with \stress have the smallest number of crossings
compared to other energy-based methods implemented in \ogdf. In order to avoid
side effects, we first computed a random drawing for each graph $G_C$ where each
coordinate is chosen uniformly at random on a grid of size $m\times m$.
Afterwards we applied the \stress method implemented in \ogdf~\cite{ogdf}
(snapshot 2017-07-23) to this drawing. 

\subparagraph{Configurations.} The previously described approach moves the
vertices in a certain order. We use the order proposed
in~\cite{doi:10.1137/1.9781611975055.12}, i.e, in descending order with respect
to the function $\Cr(\Gamma_0, v_i)^2, v_i \in V$, where $\Gamma_0$ is the
initial drawing. The computation of the new position $p^\star_i$ of a vertex
$v_i$ depends on three parameters $(|S_i|, |P_i|, K)$.
The parameter $K$ is a threshold on the degree $k_i$ of $v_i$, since  we
observed in our preliminary experiments, that in case that $k_i$ is large,
$128\,GB$ of memory are not sufficient to compute the crossing-minimal region.
Note that in case that $|S_i|$ is constant the running time to compute $R_i$ is
$O((k_i \cdot n')^2 \log n') = O(k_i^2)$, where $n' = |V(S)| \in O(|S|)$.  We
handle vertices of degree larger than $K$, as follows.  
Let $N_1\cup \dots \cup N_l$ be a partition of the neighborhood $N(v)$ of $v$
with $l = |N(v)| / K$. Further, let $u_1, u_2, \dots, u_k$ be a random order of
$N(v)$, then $N_j$ contains the vertices $u_a$ with $j \leq a \leq j + K$. For
each $j$, we compute a random sample $S_i^j$ and a crossing-minimal position
$q_j^\star$ of vertex $v$ with neighborhood $N_j$ with respect to $S^j_i$.  The
new position $p_i^\star$ of $v_i$ is the position that minimizes
$\Cr(\moveVertex{v_i}{q^\star_j}, v_i)$.

We select the same parameters for each vertex and thus denote the triple by
$(|S|, |P|, K)$.  We expect that with an increasing number $|S|$ the number of
crossings decreases.  The sample size $|S|=512$, was the largest number of
samples such that we are able to compute a final drawing of our benchmark
instances in reasonable time. As a baseline we sample $1000$ points in the
plane. Thus, we evaluate the following two configuration, $\SF=(512, 1, 100)$
and $\SN=(0, 1000, \infty)$.  Finally, we restrict the movement of a single
vertex to be within an axis-aligned square that is twice the size of the
smallest axis-aligned squares that entirely contains $\Gamma_0$.

\begin{table}
	\caption{Statistics for the \dimacs and \sparsemc graphs. $n$, $m$, and $\overline{\Delta}$
correspond the number of vertices, edges and the mean vertex-degree, respectively.}
	\label{tab:stats:dimacs:512}
	\centering
\begin{tabular}{lrrr|rrr|rr}

	&  \textbf{n}  &  \textbf{m} & $\overline{\Delta}$  &
	\multicolumn{3}{c|}{\textbf{crossings}}&  \multicolumn{2}{c}{\textbf{time
	[min]}}\\
	& & & & \stress & \SF & \SN &
	\SF & \SN \\

	\hline
	\dimacs\\
\hline
            adjnoun &   102 &    415 &     8.14 &                6\,576&  	3\,775 &       4\,468 &      0.11 &          0.09 \\
           football &   115 &    613 &    10.66 &                6\,865 &   3\,568 &       4\,030 &      0.14 &          0.17 \\
         netscience &   352 &    887 &     5.04 &                1\,724 &     583 &         814 &      0.53 &          0.31 \\
 c.metabolic &   445 &  2\,017 &     9.07 &              113\,117 &  55\,714 &      63\,028 &     11.29 &          2.29 \\
     c.neural &   282 &  2\,133 &    15.13 &              128\,068 &  86\,641 &      90\,920 &      5.23 &          2.07 \\
               jazz &   193 &  2\,737 &    28.36 &              223\,990 & 143\,647 &     153\,040 &      5.22 &          3.31 \\
              power & 3\,353 &  5\,006 &     2.99 &                7\,622 &   6\,854 &       6\,293 &      4.56 &         10.74 \\
              email &   978 &  5\,296 &    10.83 &              504\,144 & 342\,020 &     357\,272 &     37.12 &         12.48 \\
             hep-th & 4\,786 & 12\,766 &     5.33 &              836\,809 & 546\,780 &     638\,069 &     72.86 &         78.24 \\
\hline
\sparsemc \\
\hline
			1138\_bus &   671 &    991 &     2.95 &                657 &       402 &         467 &      0.41 &          0.33 \\
       ch7-6-b1 &   630 &  1\,243 &     3.95 &             64\,055 &    24\,928 &      26\,055 &      6.54 &          0.79 \\
         mk9-b2 & 1\,260 &  3\,774 &     5.99 &            412\,397 &   248\,884 &     252\,198 &     20.33 &          7.14 \\
       bcsstk08 & 1\,055 &  5\,927 &    11.24 &            455\,069 &   342\,996 &     344\,644 &     67.30 &         18.70 \\
       mahindas & 1\,258 &  7\,513 &    11.94 &          1\,463\,437 &   933\,247 &   1\,042\,787 &     68.17 &         24.09 \\
       eris1176 &   892 &  8\,405 &    18.85 &          1\,682\,458 & 1\,030\,881 &   1\,087\,605 &     77.09 &         27.33 \\
 commanche\_d & 7\,920 & 11\,880 &     3.00 &              6\,332 &     6\,239 &       6\,146 &      6.52 &         56.75 \\
\end{tabular}
\end{table}

\subparagraph*{Evaluation.}
\cref{tab:stats:dimacs:512} lists statistics for the \dimacs and the \sparsemc
graphs. In particular the number of crossings of the initial drawing (\stress)
and the drawing obtained by the \SF and \SN configurations.  Furthermore, we
report the running times for the two configurations.  Since we use an external
library (\ogdf) to compute the initial drawing, the reported times do not
include the time to compute the initial drawing. Note that \stress required at
most $0.9$\,min to complete on the \dimacs graph and $2.3$\,min on the \sparsemc
graphs. Since the size of the arrangement $\Arr(\Gamma, v)$ depends on the
degree of $v$, the overall running time varies with the number of vertices and
the average degree. Compare, e.g., c.metabolic to c.neural, or mk9-b2 to
bcsstk08. Moreover, the commanche\_d graph shows that the running time of $\SN$
is not necessarily smaller than the running time of $S_{512}$. For each point $p
\in P$ the number of crossings of edges incident to $v$ in $\moveVertex{v}{p}$
have to be counted. Since the commanche\_d graph contains over $11\,000$ edges,
the $\SF$ configuration with $|P|=1$ is faster than the $\SN$ configuration,
which has to count the number of crossings for $1\,000$ points.

Now consider the number of crossings in the initial drawing (\stress) and in the
drawing obtained by the $\SF$ configuration.  Since we move a vertex only if it
decreases its number of crossings, it is expected that the number of crossings
decreases on all graphs. For most graphs, the $\SF$ configuration decreases the
number of crossings by over $30\%$. In case of the ch7-6-b1 and the netscience
graph the number of crossings are even decreased by over $60\%$.  Exceptions are
the bcsstk08, power and commanche\_dgraphs with $24\%$, $10\%$ and $1.4\%$
respectively. 
Comparing the number crossings obtained by \SF to the configuration \SN, \SN
results in fewer crossings only on two graphs (power, commanche\_d).

\begin{figure}
	\centering
	\begin{subfigure}[b]{.3\textwidth}
		\includegraphics[page=1]{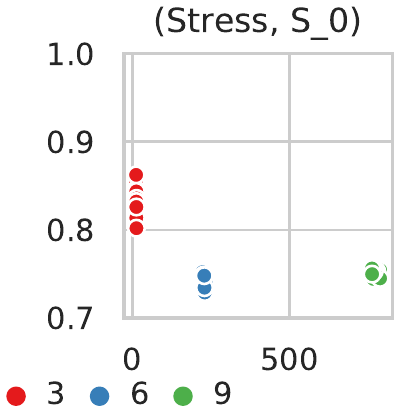}
		\caption{}
		\label{fig:k_regular:0}
	\end{subfigure}
	\quad
	\begin{subfigure}[b]{.3\textwidth}
		\includegraphics[page=2]{plots/k_regular/comp.pdf}
		\caption{}
		\label{fig:k_regular:512}
	\end{subfigure}
	\quad
	\begin{subfigure}[b]{.3\textwidth}
		\includegraphics[page=3]{plots/k_regular/comp.pdf}
		\caption{}
		\label{fig:k_regular:512_vs_0}
	\end{subfigure}

	\caption{Number of crossings of the $k$-regular graphs.}
	\label{fig:k_regular}	
\end{figure}

Observe that the power, 11138\_bus, ch7-6-b1 and commanche\_d graphs all have an
average vertex-degree of roughly $3.0$. The comparison of the number of crossing
obtained by $\SF$ and $\SN$ is not conclusive, since \SN yields fewer
crossings on the power and commanche\_d graphs and \SF on the remaining two. In
order to  be able to further study the effect of the (average) vertex degree
we evaluate the number of crossings of $k$-regular graphs.
We use the plots in \cref{fig:k_regular} for the evaluation.  Each point $(x_G,
y_G)$ corresponds to a $k$-regular graph $G$. The color encodes the
vertex-degree. Let $\Gamma_A$ and $\Gamma_B$ be two drawings of $G$ obtained by
an algorithm $A$ and $B$, respectively. The $x$-value $x_G$ corresponds to the
number of crossings in $\Gamma_A$ in thousands, i.e., $\Cr(\Gamma_A) / 1000$.
The $y$-value $y_G$ is the quotient $\Cr(\Gamma_B)/ \Cr(\Gamma_A)$. The titles
of the plots  are in the form $(A,B)$ and encode the compared algorithms. For
example in \cref{fig:k_regular:0} algorithm $A$ is \stress and $B$ is \SN.
For example, the \stress drawings of the $3$-regular graphs have on average
$12\,487$ crossings. Drawings obtained by \SN have on average $17\%$ less
crossings, i.e., $10\,402$.
On the other hand, \SF decreases the
number of crossings on average by $20\%$. For $k=6,9$, \SN and \SF both reduce
the number of crossings by $25\%$. In particular,
\cref{fig:k_regular:512_vs_0} shows that for $k=6,9$ it is unclear, whether
\SF or \SN computes drawings with fewer crossings.

\subsection{Weighted Sampling}
\label{sec:eval_weighted}

For some graphs, the previous section gives first indications that sampling a
set of edges yields a small number of crossings compared to a pure sampling of
points in the plane. In particular \cref{fig:k_regular:512_vs_0} indicates
that the edge-sampling approach does not always have a clear advantage over
sampling points in the plane. One reason for this might be that sampling within
the set of points $P_i$ in the region $R_i$ is too restrictive.  Observe that
the region $R_i$ is only crossing-minimal with respect to the sample $S$ and
does not necessarily contain the crossing-minimal position $p^\star_i$ of the
vertex $v_i$ with respect to all edges $E$.  On the other hand, sampling the set
of points $P_i$ in $\R^2$ does not use the structure of the graph at all.  This
motivates the following \emph{weighted} approach of sampling points in $\R^2$.

For a set $S \subset E$, let $\Cr_j$ be the number of crossings of the vertex
$v_i$ with respect to $\Gamma|_S$, when $v_i$ is moved to a cell $c_j$ of the
arrangement $\Arr(\Gamma|_S, v_i)$. Let $M$ be the maximum of all $\Cr_j$. We
select a cell $c_j$ with the probability $2^{M-\Cr_j} / \sum_k 2^{M-\Cr_k}$.
Within a given cell, we draw a point uniformly at random.  Note that in
case that there are exactly $n$ cells such that cell $c_j$ induces $j$
crossings, the probability that the cell $c_0$ is drawn converges to $1/2$ for
$n \rightarrow \infty$.

\subparagraph*{Benchmark Instances, Preprocessing \& Methodology.} We use the
same benchmark set and the same preprocessing steps as described in
\cref{sec:random_sampling}.  In order to obtain more reliable results, we
perform 10 independent iterations for each configuration on the \dimacs and
\sparsemc graphs. Since the $k$-regular graphs are uniform randomly computed,
they are already representative for their class. Therefore, we perform only
single runs on these graphs.

\subparagraph{Configuration.} We compare the following three configurations.
		$\RN$ refers to the uniform random sampling of points in $\R^2$ with the
	parameters $(|S|, |P|, K) = (0, 1000,\infty)$,
$\RF$ to the restricted sampling in $R_i$ with the parameters, $(512,
	1000, 100)$, and
	$\WF$ to the weighted sampling in $\R^2$ with the parameters $(512,
	1000, 100)$.
The configurations are selected such that $\RN$ and $\RF$ differ only in a
single parameter, i.e., in the number of sampled edges.  The only difference
between $\RF$ and $\WF$ is the sampling strategy. Note that the
parameters of $\RN$ and $\SN$ coincide, but not the parameters of \SF and
$\RF$.

\begin{table}
\centering
\caption{Mean and standard deviation (std) of the number of crossing categorized
by configuration. For each graph the configuration with the lowest and highest
number of crossings in marked.}
\label{tab:stats:weighted}
\begin{tabular}{lrr|rr|rr}
& \multicolumn{2}{c|}{$\RN$} & \multicolumn{2}{c|}{$\RF$} &
\multicolumn{2}{c}{$\WF$} \\
& mean & std & mean & std & mean & std \\
\hline
\dimacs\\
\hline
				adjnoun &        4\,445.0 &         39.55 &        \good 3\,655.7 &         62.96 &        3\,951.2 &         19.53 \\
           football &        \bad 3\,973.6 &         97.93 &        \good3\,350.0 & 83.38 &        \good3\,247.0 &         73.84 \\
         netscience &          \bad819.0 &         30.73 &          497.1 & 28.78 &          \good437.8 &         12.87 \\
 c.metabolic &       \bad62\,170.4 &        760.47 &       \good56\,032.3 & 1\,227.23 &       \bad62\,987.9 &      1\,907.64 \\
     c.neural &       89\,744.3 &      1\,239.22 &       \good86\,500.8 & 1\,364.5 &       \bad 99\,426.1 &      1\,258.98 \\
               jazz &      152\,013.8 &      1\,930.13 &      \good147\,387.1 & 3\,134.15 &      \bad 213\,019.4 &      1\,696.07 \\
              power &        6\,301.1 &         33.51 &        4\,512.8 & 63.09 &        \good 3\,912.5 &         30.97 \\
              email &      \bad 356\,583.4 &       3\,512.0 &      \good 341\,503.8 &      3\,480.74 &      351\,168.7 &      2\,624.18 \\
             hep-th &      \bad 640\,515.2 &      3\,443.22 &      515\,109.1 & 3\,983.23 &   \good392\,189.7 &      1\,551.53 \\
\hline
\sparsemc \\
\hline
			1138\_bus &          \bad474.6 &         13.25 &          342.9 & 12.91 &          \good247.6 &           9.8 \\
       ch7-6-b1 &       \good 25\,874.7 &        356.58 &       \good 25\,172.4 & 582.48 &       \bad28\,443.5 &         960.3 \\
         mk9-b2 &      \bad 251\,360.9 &      1\,514.05 &      245\,447.4 & 2\,914.18 &      \good 228\,794.5 &      2\,069.96 \\
				bcsstk08 &      \bad 346\,404.4 &       3\,730.3 &      \good 328\,182.0 & 6\,127.69 &  \good 330\,213.8 &      1\,726.01 \\
			mahindas &    1\,036\,745.7 &     11\,494.88 &      \good 936\,889.0 & 11\,207.34 &    \bad 1\,105\,850.9 &     10\,185.51 \\
			eris1176 &    1\,103\,184.6 &     21\,475.11 &    \good1\,037\,509.5 & 29\,877.3 &    \bad 1\,492\,423.4 &     25\,457.93 \\
 commanche\_d &        \bad6\,135.2 &         13.08 &        \good 5\,370.3 &         24.75 &        5\,979.4 &         14.72 \\
\end{tabular}
\end{table}

\subparagraph*{Evaluation.} Since we executed 10 independent runs of the
algorithm on each graph, \cref{tab:stats:weighted} lists the mean and standard
deviation of the computed number of crossings for each graph. For each graph, we
marked the cell with the lowest number of crossings in green and the largest
number of crossings in blue. For each graph, we used the Mann-Witney-U
test~\cite{sheskin2003handbook} to check the null hypothesis that the crossing
numbers belong to the same distribution. The test indicates that we can reject
the null hypothesis at a significance level of $\alpha=0.01$, for all graphs
with the exception of football, ch7-6-b1 and bcsstk08.  First, observe that the
$\RN$ configuration never computes a drawing with fewer crossings than $\RF$.
Including the football, ch7-6-b1 and the bcsstk08 graphs, $11$ of the drawings
with the fewest crossing were obtained from the $\RF$ configurations. Only $7$
correspond to the $\WF$ configuration.  \cref{tab:stats:dimacs:512} shows that
these graphs have an average vertex-degree of at most 11. Moreover,
\cref{appendix:degree_distr} shows that the degree-distributions of these graphs
follow the power-law. On the other hand, a few of the 8 graph where $\RF$
outperforms $\WF$ also have a small average vertex-degree. 

\begin{figure}
	\centering
	\begin{subfigure}[b]{.3\textwidth}
		\includegraphics[page=1]{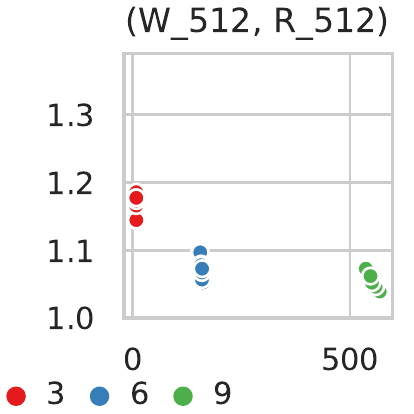}
		\caption{}
		\label{fig:apnx:distr:k_regular:w512_r512}
	\end{subfigure}
	\begin{subfigure}[b]{.3\textwidth}
		\includegraphics[page=2]{plots/k_regular/distr_scatter.pdf}
		\caption{}
		\label{fig:apnx:distr:k_regular:w512_r0}
	\end{subfigure}
	\begin{subfigure}[b]{.3\textwidth}
		\includegraphics[page=3]{plots/k_regular/distr_scatter.pdf}
		\caption{}
		\label{fig:apnx:distr:k_regular:r512_r0}
	\end{subfigure}

	\caption{Comparison of the number of crossing of the $k$-regular graphs
		computed by $\WF$ and $\RF$.}
	\label{fig:distr:k_regular}
\end{figure}

We use  \cref{fig:distr:k_regular} to compare the effect of the vertex-degree on
the number of crossings. The plot follows the same convention as the plots in
\cref{fig:k_regular}. Observe that for each $k$, the $\WF$ configuration
computes drawings with fewer crossings than $\RF$.  The improvement decreases
with an increasing~$k$. The same observation can be made for the comparison of
$\WF$ to $\RN$ but not for the comparison for $\RF$ to  $\RN$, which indicates
that sampling the set of points $P_i$ within the region $R_i$ is indeed too
restrictive, at least on our $k$-regular graphs.

Overall our experimental evaluation shows that even
with a naive uniform random sampling of a set of points in the plane the number of
crossings in drawings of \stress can be reduced considerably. Using a random
sample of a subset of the edges helps to compute drawings with even less
crossings.  The mean-vertex degree and the degree-distributions are good
indicators for whether the restrictive or the weighted sampling of the point set
$P_i$ results in a drawing with the smallest number of crossings.

\section{Conclusion}

In our previous work we showed that the primitive operation of moving a single
vertex to its crossing-minimal position significantly reduces the number of
crossings compared to drawings obtained by \stress. In this paper we introduced
the concept of \emph{bloated dual of line arrangements}, a combinatorial
technique to compute a dual representation of line arrangements. In our
applications of computing drawings with a small number of crossings, this
technique resulted in a speed-up of factor of $20$.  This improvement was
necessary to adapt the approach for graphs with a large number of vertices and
edges. On the other hand, since the worst-case running time is super-quadratic,
this improvement is not sufficient to cope with large graphs.  In
\cref{sec:random_sampling} we showed that random sampling is a promising
technique to minimize crossings in geometric drawings.  In
\cref{sec:apx_cocrossings} we proved that a random subset of edges of size
$\Theta(k \log k)$ approximates the co-crossing number of a vertex $v$ with a
high high probability.  Further, we evaluated three different strategies to
sample a set of points in the plane in order to compute a new position  for the
vertex $v_i$. First, the evaluation confirms that the number of crossings
compared to \stress can be reduced considerably. Furthermore, sampling a small
subset of the edges is sufficient to reduce the number of crossings compared to
a naive sampling of points the plane.  Our evaluation suggests that weighted
sampling is a promising approach to reduce the number of crossings in graphs
with a low average vertex degree.  Otherwise, the evaluation indicates that
restricted sampling results in fewer crossings.

The running time of the vertex-movement approach in combination with the
sampling of the edges mostly depends on the number of vertices. Since a single
movement of a vertex is not optimal anymore, two vertices can be moved
independently.  Thus, future research should be concerned with the question whether
a parallelization over the vertex set is able to further reduce the running time
while preserving the small number of crossings. Moreover, we ask whether it is
sufficient to move a small subset of the vertices to considerably reduce the
number of crossings.
 

\bibliography{strings,rcm}
\iftrue
\clearpage

\appendix

\clearpage
\section{Degree Distribution}
\label{appendix:degree_distr}

The plots in the \crefrange{fig:deg_distr:r512}{fig:deg_distr:rem} show the
degree distribution of the \dimacs and \sparsemc graphs that are listed in
\cref{tab:stats:dimacs:512}. A graph is listed in \cref{fig:deg_distr:r512} if
the configuration $\RF$ computed drawings with clearly less crossings than
$\WF$. In case that $\WF$ computes a drawing with less crossings, then the graph
is listed in \cref{fig:deg_distr:w512}. If no distinction can be made, the graph
is listed in \cref{fig:deg_distr:rem}.  Observe that all graphs in
\cref{fig:deg_distr:r512} tend to have power-law distribution. The plots in
\cref{fig:deg_distr:r512} and \cref{fig:deg_distr:rem} contains distributions
that follow the power-follow but also distributions that tend to be normal or
unstructured. 

\begin{figure}[h!]
	\begin{subfigure}[b]{.3\textwidth}
		\includegraphics[page=1]{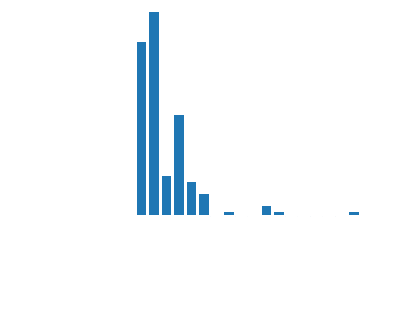}
		\caption{adjnoun}
	\end{subfigure}
	\begin{subfigure}[b]{.3\textwidth}
		\includegraphics[page=1]{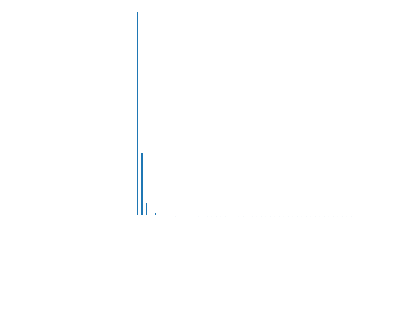}
		\caption{celegans\_metabolic}
	\end{subfigure}
	\begin{subfigure}[b]{.3\textwidth}
		\includegraphics[page=1]{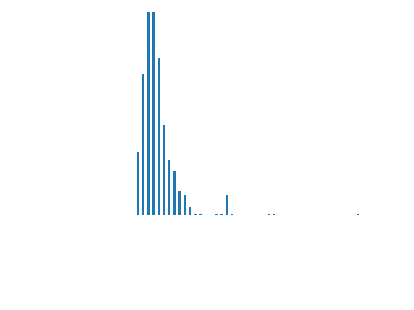}
		\caption{celegansneural}
	\end{subfigure}
	
	\begin{subfigure}[b]{.3\textwidth}
		\includegraphics[page=1]{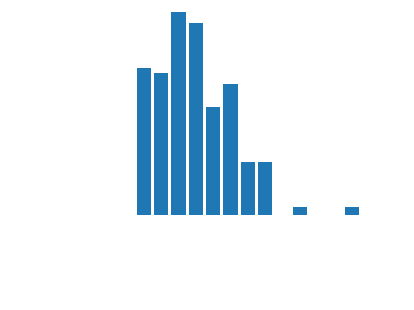}
		\caption{jazz}
	\end{subfigure}
	\begin{subfigure}[b]{.3\textwidth}
		\includegraphics[page=1]{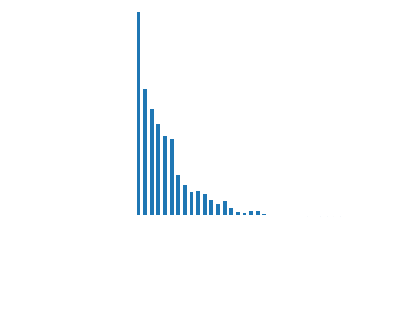}
		\caption{email}
	\end{subfigure}
	\begin{subfigure}[b]{.3\textwidth}
		\includegraphics[page=1]{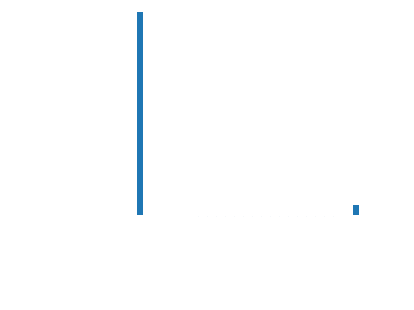}
		\caption{ch7-6-b1}
	\end{subfigure}

	\begin{subfigure}[b]{.3\textwidth}
		\includegraphics[page=1]{plots/degree/ch7-6-b1.pdf}
		\caption{commanche\_dual}
	\end{subfigure}
	\begin{subfigure}[b]{.3\textwidth}
		\includegraphics[page=1]{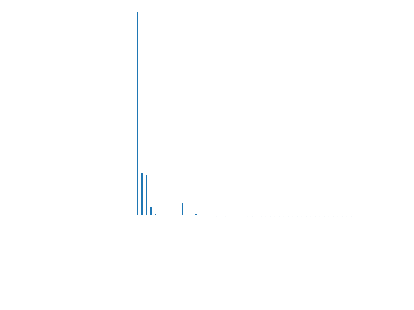}
		\caption{mahindas}
	\end{subfigure}
	\caption{Degree distribution of graphs on which the $\RF$ computes a small
number of crossings.}
\label{fig:deg_distr:r512}
\end{figure}

\begin{figure}
	\begin{subfigure}[b]{.3\textwidth}
		\includegraphics[page=1]{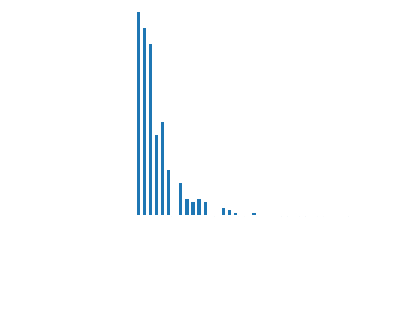}
		\caption{netscience}
	\end{subfigure}
	\begin{subfigure}[b]{.3\textwidth}
		\includegraphics[page=1]{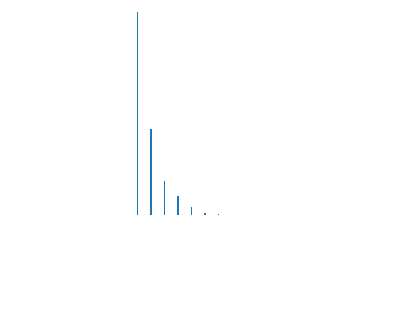}
		\caption{power}
	\end{subfigure}
	\begin{subfigure}[b]{.3\textwidth}
		\includegraphics[page=1]{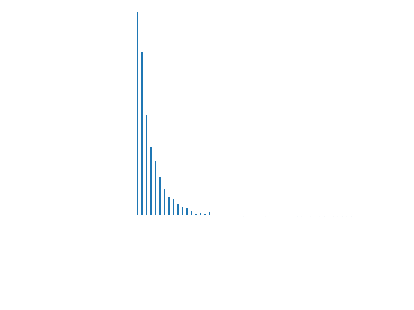}
		\caption{hep-th}
	\end{subfigure}

	\begin{subfigure}[b]{.3\textwidth}
		\includegraphics[page=1]{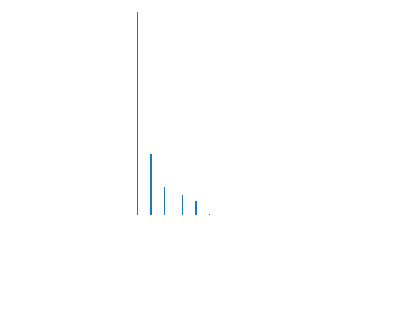}
		\caption{1138\_bus}
	\end{subfigure}
	\begin{subfigure}[b]{.3\textwidth}
		\includegraphics[page=1]{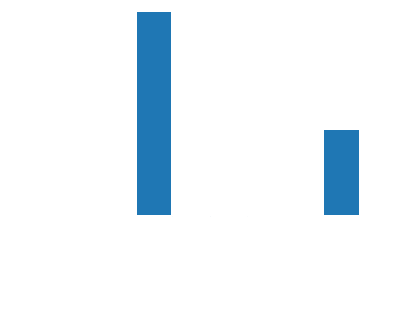}
		\caption{mk9-b2}
	\end{subfigure}
	\caption{Degree distribution of graphs on which the $\WF$ computes a small
number of crossings.}
	\label{fig:deg_distr:w512}
\end{figure}

\begin{figure}
	\begin{subfigure}[b]{.3\textwidth}
		\includegraphics[page=1]{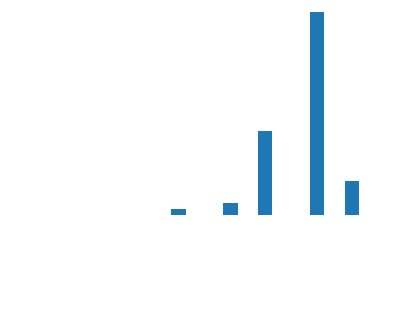}
		\caption{football}
	\end{subfigure}
	\begin{subfigure}[b]{.3\textwidth}
		\includegraphics[page=1]{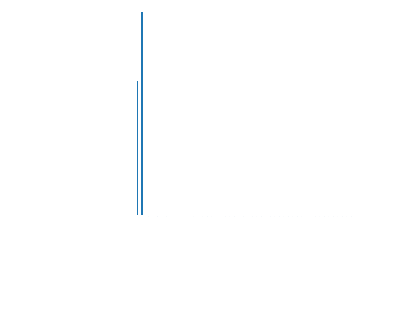}
		\caption{bcsstk08}
	\end{subfigure}
	\begin{subfigure}[b]{.3\textwidth}
		\includegraphics[page=1]{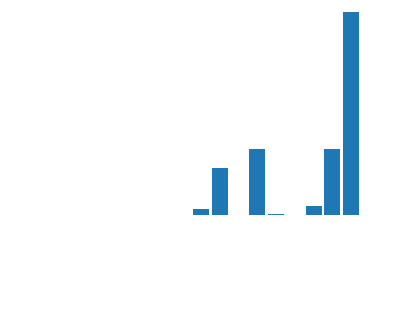}
		\caption{bcsstk27}
	\end{subfigure}
	\caption{Degree distribution of the remaining graphs.}
	\label{fig:deg_distr:rem}
\end{figure}

\clearpage
\section{Statistics of the $k$-regular}
\label{appendix:stats:k_regular}

\cref{tab:stats:k_regular} lists the statistics that correspond to the plots
in \cref{fig:k_regular}.

\begin{table}[h!]
\caption{Mean Number of crossings and standard deviation of number of crossings
in drawings of the $k$-regular graphs computed by \SN and \stress.} 
\label{tab:stats:k_regular}
\begin{subtable}[b]{.9\textwidth}
\caption{\SN vs \stress}
\begin{tabular}{rrrrrrr}
$k$ & \multicolumn{2}{l}{crossings $S_0$} & \multicolumn{2}{l}{crossings stress}\\
    &      mean &   std &             mean &   std \\
\hline
	3 &   10\,402.64 &   258.90 &         12\,487.96 &   384.04 \\
  6 &  169\,365.52 &  2260.86 &        227\,303.68 &  3450.72 \\
  9 &  580\,661.80 &  6333.13 &        774\,791.92 &  8461.29 \\
\end{tabular}
\end{subtable}

\vspace{0.5cm}
\begin{subtable}[b]{.9\textwidth}
	\caption{\SF vs \stress}
\begin{tabular}{rrrrrrr}
deg & \multicolumn{2}{l}{crossings $S_{512}$} &
\multicolumn{2}{l}{stress\_crossings} \\
    &      mean &   std &             mean &   std \\
\hline
  3 &   100\,43.76 &   285.83 &         12\,487.96 &   384.04 \\
  6 &  170\,558.48 &  2379.56 &        227\,303.68 &  3450.72 \\
  9 &  584\,505.16 &  7393.01 &        774\,791.92 &  8461.29 \\
\end{tabular}
\end{subtable}
\end{table}

\section{Missing Proofs}
\label{appendix:missing_proofs}
\obsrelativeapx*

\begin{proof}
	In order to proof the claim, we make a case distinction
	based on the size of $R$. We first assume that
		$\tsize{S \cap X} / \tsize{S} < \tsize{R}/\tsize{X}$.
	Thus, we immediately get that 
		$\tsize{X} \tsize{S \cap R} / \tsize{S} \leq |R| \leq (1+\delta) |R|$
	Moreover, the following holds
	$\tsize{\tsize{S \cap R} / \tsize{S} - \tsize{R} / \tsize{X}} =  \tsize{R} /
	\tsize{X} - \tsize{S \cap R} / \tsize{S}$. Starting from the fact $S$ is
	$(\varepsilon, \delta)$-approximation, we can do the following
	transformations.
	\begin{align*}
		& \size{X} \left(\frac{\size{R}}{\size{X}} - \frac{\size{S \cap R}}{
			\size{S}}\right) \leq
			\delta \tsize{X} \max\left\{ \frac{\size{R}}{\size{X}}, \varepsilon
		\right\} \\
		\Leftrightarrow & \size{R} - \size{X} \frac{\size{S \cap R}}{\size{S}} \leq
			\delta \max\left\{\size{R}, \varepsilon \size{X}  \right\}\\
		\Leftrightarrow & \size{X} \frac{\size{S \cap R}}{\size{S}} \geq
		\size{R}  -\delta \size{R} = (1-\delta) \size{R} 
	\end{align*}

	In order to complete the proof, assume that
		$\tsize{S \cap X} / \tsize{S} \geq \tsize{R}/\tsize{X}$. 
	\begin{align*}
		& \size{X} \left(\frac{\size{S \cap R}}{
			\size{S}}  - \frac{\size{R}}{\size{X}}\right) \leq
			\delta \tsize{X} \max\left\{ \frac{\size{R}}{\size{X}}, \varepsilon
		\right\} \\
		\Leftrightarrow & \size{R} - \size{X} \frac{\size{S \cap R}}{\size{S}} \leq
		\delta \max\left\{\size{R}, \varepsilon \size{X}  \right\}\\
		\Leftrightarrow & \size{X} \frac{\size{S \cap R}}{\size{S}} \leq
		\size{R} + \delta \size{R} = (1+\delta) \size{R}%
	\end{align*}%
\end{proof}

\lemmaapxsum* \begin{proof}
	Recall that $\coCr(p)$ is equal to $\sum_{u \in N(v)} \tsize{
		\coCrEdge{uv}{p} }$.  Since the drawing~$\Gamma$ is $\varepsilon$-well
		behaved, for every $u \in N(v)$ and every $p \in \R^2$ we have that at least
		an $\varepsilon$-fraction of edges is not crossed by the edge $uv$, i.e.,
		$\tsize{\coCrEdge{uv}{p}} \geq \varepsilon|E|$.  Since $S_u$ is a relative
		$(\varepsilon, \delta)$-approximation and due to
		Proposition~\ref{obs:relative_apx}  we have that $ (1-\delta)
		|\coCrEdge{uv}{p}| \leq |E| |\coCrEdge{uv}{p} \cap S_u| / |S_u| \leq
		(1+\delta) |\coCrEdge{uv}{p}|$. Plugging this inequality into the sum of
		$\lambda(p)$ proves the lemma.

\end{proof}

\lemmarandomsample*
\begin{proof}
	For each vertex $u \in N(v)$, we denote with $A_u$ the event that $S$ is a
	relative $(\varepsilon,
	\delta)$-approximation of the set system $(E, \co{ \mathcal F_{uv}})$.
	According to 
	Lemma~\ref{lemma:vc_dimension} and Theorem~\ref{theorem:rel_eps_approx} the
	probability $\Prob(A_u)$ that a
	uniformly random sample is a relative $(\varepsilon, \delta)$-approximation of
	$(E, \co{\mathcal F_{uv}})$ is $1-\gamma$.
	The following estimate can be proven by induction using the equalities
		$\Prob(A \wedge B)  = \Prob(A) + \Prob(B) - \Prob(A \vee B)$
	and
		$\Prob(A \vee B) \leq 1$. 

	\[ \Prob\left( \bigwedge_{u \in N(v)} A_u \right) \geq \sum_{u \in N(v)}
	\Prob(A_u) - k + 1 \] 

	Plugging in the probability for $\Prob(A_u)$ proves that 
	$S$ is a relative $(\varepsilon,
	\delta)$-approximation with probability $1 - k\gamma$ for a $\gamma \leq
	1/ k$.
\end{proof}
\fi
\end{document}